\documentclass[11pt,letterpaper,preprintnumbers,superscriptaddress]{revtex4}
\usepackage{amsfonts,amsmath,amsthm,amssymb,amscd,dsfont}
\usepackage{color,graphicx,tikz,array}
\usepackage{hyperref,cleveref}
\normalfont
\usepackage[T1]{fontenc}

\newcommand{\bra}[1]{\ensuremath \langle{#1}|}%
\newcommand{\ket}[1]{{\ensuremath |{#1}\rangle}}%

\newtheorem{theorem}{Theorem}
\newtheorem{lemma}[theorem]{Lemma}
\newtheorem{proposition}[theorem]{Proposition}

\crefname{inequality}{ineq.}{ineqs.}
\creflabelformat{inequality}{(#2#1#3)}
\crefname{definition}{}{}
\creflabelformat{definition}{(#2#1#3)}

\definecolor{green}{rgb}{0.1,0.7,0.1}

\begin{document}

\title{Penalty models for bitstrings of constant Hamming weight}
\date{\today}

\author{Brad~Lackey}
\affiliation{Joint Center for Quantum Information and Computer Science, University of Maryland, College Park}
\affiliation{Departments of Computer Science and Mathematics, University of Maryland, College Park}
\affiliation{Mathematics Research Group, National Security Agency, Ft.~G.~G.~Meade, Maryland}

\begin{abstract}
To program a quantum annealer, one must construct objective functions whose minima encode hard constraints imposed by the underlying problem. For such ``penalty models,'' one desires the additional property that the gap in the objective value between such minima and states that fail the constraints is maximized amongst the allowable objective functions. In this short note, we prove the standard penalty model for the constraint that a bitstring has given Hamming weight is optimal with respect to objective value gap.
\end{abstract}

\maketitle

%\tableofcontents

\section{Introduction and preliminaries}

When developing algorithms for constrained optimization problems, a common task is ``mapping.'' By this we mean the conversion of a hard constraint into an objective function, or ``penalty model.'' The key property is that configurations satisfying the constraint are precisely those minimizing this objective function. An important secondary property is that the penalty model be optimal, in that the minimum penalty assigned to any non-satisfying configuration is as large as possible, within the class of objective functions under consideration. This allows the largest dynamic range for further optimization tasks among the satisfying configurations. As opposed to our previous work \cite{bian2014discrete,bian2016mapping}, we focus here on quadratic unconstrained binary objective functions (QUBOs). A general QUBO on $n$ bits has the form
\begin{equation}\label[definition]{eqn:qubo}
Q(x_1, \dots, x_n) = a + \sum_{j=1}^n b_jx_j + \sum_{1\leq j<k\leq n} c_{jk}x_jx_k,
\end{equation}
where $a,b_j,c_{jk} \in \mathds{R}$. 

A frequent task when mapping planning/scheduling problems is penalizing multiple actions. In these type of problems it is customary to use a binary variable $x_{i,a}$ to represent whether task $i$ has been assigned to agent $a$, \cite{rieffel2015case,venturelli2015quantum}; one has a constraint that for each $i$, we have $x_{i,a} = 1$ for precisely one $a$. This sort of constraint is common to coloring/covering problems \cite[\S{6}]{lucas2013ising}, and also can be found in problems where incidence constraints must be enforced \cite[\S\S{7-9}]{lucas2013ising}. A generalization of this can be found in graph partition and clique finding problems where precisely $r$ bits must be set to one, \cite{childs2000finding,lackey2016partitioning} and \cite[\S{2}]{lucas2013ising}.

To illustrate, consider the task of developing a penalty model for the constraint that a bitstring have Hamming weight one, as in the previously mentioned planning/scheduling problems. The ``standard'' objective function is the QUBO
$$Q_1(\vec{x}) = E\cdot\left(1 - \sum_{j=1}^n x_j\right)^2 = E - E\sum_j x_j + 2E\sum_{j < k} x_jx_k,$$
where $E$ is an appropriate energy scale. Indeed, this has the desired key property: if $|\vec{x}| = 1$ then $Q_1(\vec{x}) = 0$, while if $|\vec{x}| \not= 1$ then $Q_1(\vec{x}) \geq E$. Turning to the second property, we see that the minimum energy penalty for a non-satisfying bitstring is $E$; we question if this is optimal. \textit{\'A priori} $E$ is arbitrary, so in order to discuss optimality we need to restrict to QUBOs with bounds on their coefficients. This is a somewhat subtle point, which we will return to later when we examine converting QUBOs to Ising Hamiltonians. For now let us suppose bounds on the coefficients of \cref{eqn:qubo} as $|b_j| \leq B$ and $|c_{jk}|\leq C$. Note that since only energy differences are well defined, bounds on $a$ are not realistic. So we have two cases:
\begin{enumerate}
\item if $\frac{B}{C} \leq \frac{1}{2}$ then the minimum penalty is $B$,
\item if $\frac{B}{C} \geq \frac{1}{2}$ then the minimum penalty is $\frac{C}{2}$.
\end{enumerate}

Similarly, a penalty model isolating bitstring of Hamming weight $r$ is
$$Q_r(\vec{x}) = E\cdot\left(r - \sum_{j=1}^n x_j\right)^2 = r^2E - (2r-1)E\sum_j x_j + 2E\sum_{j < k} x_jx_k.$$
Using the bounds as above, we again have two cases:
\begin{enumerate}
\item if $\frac{B}{C} \leq \frac{2r-1}{2}$ then the minimum penalty is $\frac{B}{2r-1}$,
\item if $\frac{B}{C} \geq \frac{2r-1}{2}$ then the minimum penalty is $\frac{C}{2}$.
\end{enumerate}
While \textit{ad hoc}, it turns out these QUBOs are optimal, as we will prove next.

\section{Optimality}\label{sec:optimality}

Again we separate out the case $r = 1$ since is it somewhat different. Let us write $\vec{\delta}_j$ for the bitstring that is $1$ at position $j$, but zero elsewhere. Let us write $g$ for the largest minimum penalty achievable by a QUBO on bitstrings of Hamming weight not equal $1$. Continuing the notation of \cref{eqn:qubo} we evaluate
$$\begin{array}{rcccl}
g &\leq& Q(\vec{0}) &=& a\\
0 &=& Q(\vec{\delta}_j) &=& a + b_j\\
g &\leq& Q(\vec{\delta}_j+\vec{\delta_k}) &=& a + b_j + b_k + c_{jk}.
\end{array}$$
Subtracting the middle equality from the top inequality gives
$$g \leq -b_j \leq B.$$
Now if $\frac{B}{C} \leq \frac{1}{2}$ (Case 1 above) then $Q_1$ already saturates this bound, and hence is optimal in this case.

On the other hand, if $\frac{B}{C} \leq \frac{1}{2}$ (Case 2 above), then we exploit the inequality
$$g \leq -a + (a + b_j) + (a + b_k) + c_{jk} \leq -g + c_{jk} \leq -g + C.$$
That is $2g \leq C$ and again $Q_1$ saturates this bound. Therefore $Q_1$ is optimal in every case.

For the general case of weight $r$ bitstrings, we again focus on the value of $Q$ on bitstrings of weight $r-1$, $r$, and $r+1$. Let $S \subset \{1, \dots, n\}$ be a subset of cardinality $|S| = r-1$, $S'$ a set with $|S'| = r$, and $S''$ with $|S''| = r+1$. Then evaluating $Q$ the bitstrings with these support sets produces
\begin{align}
g &\leq a + \sum_{j \in S} b_j + \sum_{\substack{j<k\\j,k\in S}} c_{jk}\label[inequality]{eqn:lower}\\
0 &=  a + \sum_{j \in S'} b_j + \sum_{\substack{j<k\\j,k\in S'}} c_{jk}\label{eqn:middle}\\
g &\leq a + \sum_{j \in S''} b_j + \sum_{\substack{j<k\\j,k\in S''}} c_{jk}.\label[inequality]{eqn:upper}
\end{align}
First, take any $S$ and $u \not\in S$ and form $S' = S\cup\{u\}$. Then using $S'$ in \cref{eqn:middle}, we subtract this from \cref{eqn:lower}, which gives us
\begin{equation}\label[inequality]{eqn:lower_bound}
g \leq -b_u - \sum_{j\in S} c_{ju},
\end{equation}
Next, take any $S'$ and $u\not\in S'$ and form $S'' = S'\cup\{u\}$. Using $S''$ in \cref{eqn:upper}, we subtract from this \cref{eqn:middle}, yielding
\begin{equation}\label[inequality]{eqn:upper_bound}
g \leq b_u + \sum_{j\in S'} c_{ju}.
\end{equation}

To obtain one bound, let us add \cref{eqn:lower,eqn:upper} with $S' = S \cup \{v\}$ (with $u\not=v$). This gives $2g \leq c_{uv} \leq C$, and so if $\frac{B}{C} \geq \frac{2r-1}{2}$ (Case 2), then $Q_r$ already has minimum penalty $\frac{C}{2}$ and is therefore optimal. 

When $\frac{B}{C} \geq \frac{2r-1}{2}$ (Case 1), we obtain a different bound by taking $u$ and $S'$ as in \cref{eqn:upper}. Now for the same $u$, we sum \cref{eqn:lower} over all $S \subset S'$ with $|S| = r-1$. Each $j \in S'$ appears in $r-1$ subsets $S \subset S'$, namely only when $\{j\} = S'\setminus S$ is it absent. Hence each $c_{ju}$ appears in $(r-1)$ inequalities, and we obtain
$$rg \leq -rb_u - (r-1)\sum_{j\in S'} c_{ju}.$$
Adding to this inequality $(r-1)$ times \cref{eqn:upper}, one gets
$$(2r-1)g \leq -b_u \leq B.$$
Again $Q_r$ saturates this inequality and so is optimal in this case as well. 

Note that only the lower bound on $b_j$ and upper bound on $c_{jk}$ were relevant to this argument. Hence we have proven the following result.

\begin{theorem}\label{thm:optimalQUBO}
Among all QUBOs whose linear coefficients are lower bounded by $-B$ and quadratic coefficients upper bounded by $C$, the QUBO $Q_r$ realizes the optimal minimal penalty model for Hamming weight $r$ bitstrings. The optimal minimal penalty is $\frac{B}{2r-1}$ when $B \leq \frac{2r-1}{2}\cdot C$ or $\frac{C}{2}$ when $B \geq \frac{2r-1}{2}\cdot C$.
\end{theorem}

\section{Restricting the topology}\label{sec:restriction}

As we have seen above, the QUBOs $Q_r$ are optimal in terms of maximizing the penalty for bitstrings not of the desired weight. However, they suffer from the fact that their interaction graph (the graph that has an edge $(j,k)$ for each nonzero $c_{jk}$) is complete. One may be willing to accept suboptimal QUBOs for a sparser graph. We see that this is not possible.

Again we separate out the case $r=1$ as this is special. Suppose $(j,k)$ is not an edge of the interaction graph of a QUBO $Q$, and so $c_{jk}=0$. Then as before,
$$\begin{array}{rclcl}
g &\leq& Q(\vec{0}) &=& a,\\
0 &=& Q(\vec{\delta}_j) &=& a + b_j.
\end{array}$$
Now evaluating $Q(\vec{\delta}_j + \vec{\delta_k})$ we obtain
$$g \leq a + b_j + b_k = -a \leq -g.$$
But $g \geq 0$ and so $g = 0$. The general case of $r > 1$ is similar, which for formally state now.

\begin{proposition}
Let $Q$ be any QUBO whose interaction graph is not complete, and has $Q(\vec{x}) = 0$ for all Hamming weight $r$ bitstrings. Then there exists a bitstring $\vec{x}$ of Hamming weight $|\vec{x}| \not= r$ so that $Q(\vec{x}) = 0$.
\end{proposition}
\begin{proof}
Suppose $c_{uv} = 0$, which exists by hypothesis, and $S$ be any set of $r-1$ indices not containing $u$ or $v$. Write $g = \min\{Q(\vec{x}) \::\: |\vec{x}| = r\pm 1\}$. Taking $S_1 = S \cup \{u\}$ we apply \cref{eqn:lower},
$$g \leq -b_u - \sum_{j\in S} c_{ju}.$$
Now taking $S_2 = S\cup \{v\}$ and $S' = S_2 \cup \{u\}$ we apply \cref{eqn:upper},
$$g \leq b_u + \sum_{j\in S_2} c_{ju}.$$
Adding these gives $2g \leq c_{uv} = 0$, and so as above $g = 0$.
\end{proof}

We hasten to indicate that this theorem does not prohibit quadratic penalty models on sparse graphs, but rather states that to produce one requires the graph have more $n$ vertices. Clearly, one can apply standard minor embedding techniques \cite{choi2008minor,choi2011minor} to $Q_r$ to accomplish this.

\section{QUBOs versus Ising Hamiltonians}

At this point, we indicate that the proof given in \cref{sec:optimality} could have been simplified by first proving that any optimal QUBO must be symmetric under reordering its variables. Unfortunately this would not have been true in the context of \cref{sec:restriction} where one or more the coupling coefficients is assumed to vanish, and so we opted for a direct proof of \cref{thm:optimalQUBO}. In this section we aim to prove an analogue of \cref{thm:optimalQUBO} for Ising Hamiltonians, and so first prove this reduction to the symmetric case.

\begin{lemma}
Let $$H = E_0 + \sum_{j=1}^n h_j s_j + \sum_{j<k} J_{jk}s_js_k$$ be an Ising Hamiltonian for spins $s_j \in \{\pm 1\}$, and $G$ be a permutation group acting on the spin indexes. Suppose we have the following properties:
\begin{enumerate}
\item the ground state manifold, $M$, of $H$ has energy $0$ and is invariant under the action of $G$;
\item any bounds on the biases and interaction coefficients are invariant under $G$ (e.g. $-h_\text{min} \leq h_j \leq h_\text{max}$ and $-J_\text{min} \leq J_{jk} \leq J_\text{max}$); and,
\item the spectral gap, $\gamma(H)$, between the ground states and first excited states is maximal among all Hamiltonians with ground states $M$ and given coefficient bounds.
\end{enumerate}
Then the exists a Hamiltonian with ground state manifold $M$, which satisfies the given coefficient bounds, achieves the maximal spectral gap, and the and takes the form
\begin{equation}\label{eqn:symmetric} 
\bar{H} = E_0 + \sum_\mu \bar{h}_\mu f_\mu(\vec{s}) + \sum_\nu \bar{J}_\nu g_\nu(\vec{s}),
\end{equation}
where $\{f_\mu\}$ and $\{g_\nu\}$ form a basis of the linear and quadratic $G$-invariant polynomials.
\end{lemma}
\begin{proof}
Let $G$ act on $H$ in the obvious way:
$$g\cdot H = E_0 + \sum_{j=1}^n h_j s_{g(j)} + \sum_{j<k} J_{jk} s_{g(j)}s_{g(k)}.$$
Then write $\bar{H} = \frac{1}{|G|}\sum_{g\in G} g\cdot H$. Then $\bar{H}$ is precisely the projection of $H$ onto the space of $G$-invariant polynomials, and so is of the form of \cref{eqn:symmetric} (see for example \cite[Chapter 4]{olver1999classical}).

Note that from property (1), the ground states of $g\cdot H$ coincide with those of $H$, and have ground state energy zero. Consequently, we have (i) $\bra\psi\bar{H}\ket\psi \geq 0$ for any $\ket\psi$, and (ii) any ground state of $H$ is a zero-energy state of $\bar{H}$. So any state in $M$ is also a ground state of $\bar{H}$.  Conversely if $\bra\psi\bar{H}\ket\psi > 0$, the for some $g$ we must have $\bra\psi (g\cdot H)\ket\psi > 0$ and hence $\ket\psi \not\in M$. Therefore, the ground state manifold of $\bar{H}$ is also $M$.

To bound the coefficients of $\bar{H}$, we note that the coefficient of $s_j$ in $\bar{H}$ is $\frac{1}{|G|}\sum_{g\in G} h_{g^{-1}(j)}$. But by property (2), each $h_{g^{-1}(j)}$ satisfies the same bounds as $h_j$, say $a \leq h_j \leq b$. Then we have
$$a = \frac{1}{|G|}\sum_{g\in G} a \leq \frac{1}{|G|}\sum_{g\in G}h_{g^{-1}(j)} \leq \frac{1}{|G|}\sum_{g\in G} b = b.$$
The same argument shows the other coefficients of $\bar{H}$ satisfies the same bounds as those of $H$.

Finally, to show that $\bar{H}$ achieves the maximal spectral gap, we argue that the spectral gaps of each $g\cdot H$ coincide with that of $H$. But this is clear: if $\ket\psi$ is an eigenstate of $H$ with energy $E$, then $g\cdot\ket\psi$ is an eigenstate of $g\cdot H$ with the same energy. And so for any first excited state of $\bar{H}$, say $\bar{H}\ket\psi = \gamma(\bar{H})\ket\psi$, we have $\ket\psi$ is orthogonal to $M$ and thus $\bra\psi (g\cdot H) \ket\psi \geq \gamma(H)$. Thus 
$$\gamma(\bar{H}) = \bra\psi \bar{H} \ket\psi =  \frac{1}{|G|} \sum_{g\in G} \bra\psi (g\cdot H) \ket\psi \geq \gamma(H).$$ 
But $\bar{H}$ is a Hamiltonian that satisfies (1) and (2) and so from property (3), $\gamma(\bar{H}) = \gamma(H)$ as $\gamma(H)$ is maximal.
\end{proof}

\begin{theorem}
Among Ising Hamiltonians whose linear coefficients $h_j$ are bounded $-h_{min} \leq h_j \leq h_{max}$ ($h_{min},h_{max} > 0$) and quadratic coefficients $J_{jk}$ are bounded $J_{jk} \leq J_{max}$, the optimal penalty model for Hamming weight $r$ bitstrings is given by
\begin{equation}\label{eqn:optimal-Ising}
H_r(s) = E\left(\tfrac{n}{2} - \tfrac{3}{2}(n-2r)^2 + (n-2r)\sum_{j=1}^n s_j + \sum_{j<k} s_js_k\right),
\end{equation}
with minimal penalty $2E$. The optimal choice of $E$ is given as follows.
\begin{enumerate}
\item If $r = \frac{n}{2}$ then $E = J_{max}$.
\item If $r < \frac{n}{2}$ then $E$ is the lesser of $J_{max}$ and $\frac{h_{max}}{n-2r}$.
\item If $r > \frac{n}{2}$ then $E$ is the lesser of $J_{max}$ and $\frac{h_{min}}{r-2n}$.
\end{enumerate}
\end{theorem}
\begin{proof}
The ground state of the Hamiltonian corresponds to all bitstrings of Hamming weight $r$, which is invariant under the entire symmetric group, $G = S_n$, as are the proscribed bounds. The $S_n$-invariant polynomials are generated by the elementary symmetric polynomials (see for example \cite[Theorem 4.23]{olver1999classical}), and so by the lemma the maximal spectral gap is achieved by a Hamiltonian of the form
$$H(s) = E_0 + h\cdot p_1(\vec{s}) + J\cdot p_2(\vec{s}) = E_0 + h\cdot\sum_{j=1}^n s_j + J\cdot \sum_{j<k} s_j s_k.$$
We hasten to point to point out the space of $S_n$-invariant quadratic polynomials is two dimensional, spanned by $p_2$ and $p_1^2$. However when restricted to spins,
$$p_1^2(\vec{s}) = \left(\sum_{j=1}^n s_j\right)^2 = \sum_{j=1}^n s_j^2 + 2\sum_{j<k}s_js_k = n + 2 p_2(\vec{s}),$$
and hence we can incorporate the contributions from $p_1^2$ into other terms.

Working from this form, we evaluate $H$ on string of weight $r-1$, $r$, and $r+1$ and find the optimal gap $g$ must satisfy
\begin{align*}
    g &\leq -h - 2J(2r-1-n)\\
    g &\leq h + 2J(2r+1-n),
\end{align*}
with all other weights providing less stringent inequalities. Adding these gives $g \leq 2J$, which the Hamilton (\ref{eqn:optimal-Ising}) saturates. The three cases ensure the bounds on the coefficients are all satisfied.
\end{proof}

\bibliography{main}

\end{document}